\documentclass[a4paper,UKenglish,cleveref, autoref, thm-restate]{lipics-v2021}

\usepackage[utf8]{inputenc}
\usepackage{graphicx}
\usepackage{booktabs}
\usepackage{amsmath}
\usepackage{amssymb}
\usepackage{bbm}
\usepackage{tikz}
\usepackage{tikz-qtree}
\hypersetup{
    colorlinks=true,
    linkcolor=blue,
    filecolor=blue,
    urlcolor=blue,
    allcolors=blue
}
\usepackage{relsize}
\usepackage{rotating}
\usepackage{algorithmicx}
\usepackage{algorithm}
\usepackage[noend]{algpseudocode}
\usepackage{multicol}
\usepackage[most]{tcolorbox}
\usepackage{color}
\usepackage{pgfplots}

\usetikzlibrary{arrows}

\pgfplotsset{
    compat=1.3,
    legend image code/.code={
        \draw [#1] (0cm,-0.1cm) rectangle (0.6cm,0.1cm);
    },
}

\usepackage{enumerate}
\usepackage{mathtools}
\usepackage{xspace}
\usepackage{svg}

\crefname{theorem}{Theorem}{Theorems}
\Crefname{lemma}{Lemma}{Lemmas}
\Crefname{claim}{Claim}{Claims}
\Crefname{observation}{Observation}{Observations}
\Crefname{algorithm}{Algorithm}{Algorithms}
\Crefname{myalgctr}{Algorithm}{Algorithms}
\Crefname{challenge}{Challenge}{Challenges}

\algrenewcommand\algorithmicindent{1em}%

\DeclarePairedDelimiter\floor{\lfloor}{\rfloor}
\usepackage{xcolor}

\definecolor{mygreen}{RGB}{20,140,80}
\definecolor{linkcolor}{RGB}{0,0,230}
\definecolor{mylightgray}{RGB}{230,230,230}
\definecolor{verylightgray}{RGB}{245,245,245}

\newcommand{\opt}{\textsc{OPT}\xspace}

\newcommand{\eps}{\varepsilon}

\newcommand{\arccoloring}{\textsc{Arc Coloring}\xspace}

\newcommand{\threepartition}{\textsc{3-partition}\xspace}

\newcommand{\CAUTCSW}{\textsc{CAUTC-SW}\xspace}
\newcommand{\CAUTCDECISION}{\textsc{CAUTC-Decision}\xspace}

\newcommand{\geom}{\mathsf{Geom}}

\newcommand{\fact}{\log^2 n}

\newcommand{\mbmwc}{\textsc{MbMwC}\xspace}
\newcommand{\mis}{\textsc{MIS}\xspace}

\newcommand{\draw}{\geom(\exp(\eps/\fact))}

\newcommand{\starttime}{start}
\newcommand{\etime}{end}
\newcommand{\valset}{U}
\newcommand{\utility}[2]{\ensuremath{u_{#1}(#2)}}
\newcommand{\scalgo}{\cref{alg:quanquan}\xspace}
\newcommand{\efonealgo}{\cref{alg:augmented-ef1}\xspace}

\newcommand{\coursejcapacity}[1]{\ensuremath{D_{#1}}}

\newcommand{\studenticreditcap}[1]{\ensuremath{c_{#1}}}

\newcommand{\defn}[1]{\textbf{\emph{#1}}}

\title{An Algorithmic Approach to Address Course Enrollment Challenges} 

\author{Arpita Biswas}{Harvard University, USA \and \url{https://sites.google.com/view/arpitabiswas} }{arpitabiswas@g.harvard.edu}{https://orcid.org/0000-0002-5720-013X}{}
\author{Yiduo Ke}{Northwestern University, USA \and \url{https://sites.northwestern.edu/yiduoke/}}{yiduoke2026@u.northwestern.edu}{https://orcid.org/0009-0000-8118-948X}{}
\author{Samir Khuller}{Northwestern University, USA\and \url{https://www.mccormick.northwestern.edu/research-faculty/directory/profiles/khuller-samir.html}}{samir.khuller@northwestern.edu}{https://orcid.org/0000-0002-5408-8023}{}
\author{Quanquan C. Liu}{Northwestern University, USA\and \url{https://quanquancliu.com/}}{quanquan@northwestern.edu}{https://orcid.org/0000-0003-1230-2754}{}

\authorrunning{A. Biswas, Y. Ke, S. Khuller, Q. C. Liu} 

\Copyright{Arpita Biswas, Yiduo Ke, Samir Khuller, Quanquan C. Liu} 

\ccsdesc[500]{Theory of computation~Scheduling algorithms}

\keywords{fairness, allocation, matching, algorithms} 

\category{} 

\relatedversion{} 

\acknowledgements{We thank Prahlad Narasimhan Kasthurirangan for pointing us to a relevant reference for one of the problems we consider. Samir Khuller and Yiduo Ke gratefully acknowledge support from NSF-Award 2216970 (IDEAL Institute).}

\EventEditors{Kunal Talwar}
\EventNoEds{1}
\EventLongTitle{4th Symposium on Foundations of Responsible Computing (FORC 2023)}
\EventShortTitle{}
\EventAcronym{}
\EventYear{2023}
\EventDate{June 7--9, 2023}
\EventLocation{Stanford University, California, USA}
\EventLogo{}
\SeriesVolume{256}
\ArticleNo{8}

\begin{document}

\maketitle

\begin{abstract}
Massive surges of enrollments in courses have led to a crisis in several computer science departments - not only is the demand for certain courses extremely high from majors, but the demand from non-majors is also very high. Much of the time, this leads to significant frustration on the part of the students, and getting seats in desired courses is a rather ad-hoc process. One approach is to first collect information from students about which courses they want to take and to develop optimization models for assigning students to available seats in a fair manner. What makes this problem complex is that the courses themselves have time conflicts, and the students have credit caps (an upper bound on the number of courses they would like to enroll in). We model this problem as follows. We have $n$ agents (students), and there are ``resources'' (these correspond to courses). Each agent is only interested in a subset of the resources (courses of interest), and each resource can only be assigned to a bounded number of agents (available seats). In addition, each resource corresponds to an interval of time, and the objective is to assign non-overlapping resources to agents so as to produce ``fair and high utility'' schedules.

In this model, we provide a number of results under various settings and objective functions. Specifically, in this paper, we consider the following objective functions: total utility, max-min (Santa Claus objective), and envy-freeness. The total utility objective function maximizes the sum of the utilities of all courses assigned to students. The max-min objective maximizes the minimum utility obtained by any student. Finally, envy-freeness ensures that no student envies another student's allocation.
Under these settings and objective functions, we show a number of theoretical results. Specifically, we show that the course allocation under the time conflicts problem is NP-complete but becomes polynomial-time solvable when given only a constant number of students \emph{or} all credits, course lengths, and utilities are uniform. Furthermore, we give a near-linear time algorithm for obtaining a constant $ 1/2$-factor approximation for the general maximizing total utility problem when utility functions are binary. In addition, we show that there exists a near-linear time algorithm that obtains a $1/2$-factor approximation on total utility and a $1/4$-factor approximation on max-min utility when given uniform credit caps and uniform utilities. For the setting of binary valuations, we show three polynomial time algorithms for $1/2$-factor approximation of total utility, envy-freeness up to one item, and a constant factor approximation of the max-min utility value when course lengths are within a constant factor of each other. Finally, we conclude with experimental results that demonstrate that our algorithms yield high-quality results in real-world settings. 
\end{abstract}

\section{Introduction}

This work addresses a central  problem in fair resource allocation in the course allocation setting. In the algorithms community, one of the fairness objectives is to allocate resources among agents to maximize the minimum allocation to any single agent, also known as ``Santa Claus'' problem. In the course allocation setting, there are additional constraints to the Santa Claus problem, such as a ``conflict'' graph between the resources, in other words, if there is a conflict edge between two resources, then we cannot allocate that pair of resources to the same agent. Our study was motivated by the course allocation scenario since massive surges in enrollments in CS courses have led to a crisis in several computer science departments - not only is the demand for certain courses extremely high from majors, but the demand from non-majors is also very high. Much of the time, this leads to significant frustration on the part of the students who are unable to get into courses of interest, and this lead to non-uniformity in student happiness as a few students were able to successfully petition faculty to add them to their course, and other students failed to get into any course of interest (leading to further annoyance when finding out that you did not get in, but your friend did).  As registration opens up, there is always a mad scramble to enroll in courses.  Given the amount of money spent by students on fees, and due to the scale of the problem, we set out to collect the information from students about which courses they want to take, and then developed optimization models for assigning students to available seats. What makes this problem complex is that courses themselves have time conflicts, so a student might be interested in two courses, but if they meet at overlapping times, they can only take one of those courses. Moreover, students have credit caps, that limit how many courses a student can enroll in, and naturally, courses have limited capacity. Students specify a set of courses that they are interested in, and we care about total utility (assigned seats), as well as fairness measured by both the lowest allocation to any student in an assignment and envy-freeness. 

While our motivating example was assigning seats to students in a fair manner, this is a pretty general resource allocation problem with some additional constraints capturing conflicts among courses and capacity constraints of students. We represent the conflict using a \textit{conflict graph} where resources are the nodes and an edge between two resources implies that those two resources cannot be assigned to the same student. 


The problem when the conflict graph is unrestricted is NP-hard (\cref{b-match}). 
Thus, we focus on the case of assigning resources that can be represented as intervals. Each interval has a start and end time. We assume that time occurs in discrete integer time steps in increments of $1$ beginning with step $0$. Overlapping intervals are those that strictly overlap (an interval ending at time 3 does not overlap with another interval that starts at 3). The conflict graph is now determined by the overlapping structure: if two resources (intervals) overlap in time, then there is an edge between them in the corresponding conflict graph.


\subsection{Related Work}\label{sec:related}
The problem of allocating resources among a set of $n$ agents with an egalitarian objective (maximizing the total value of items allocated to the worst-off agent) has been well-studied in the literature and is known as the Santa Claus problem. This problem was introduced by Bansal and Sviridenko~\cite{BansalSviridenko06} and they developed a $O(\log\log{n}/\log{} \log{} \log{n})$ approximation algorithm. Later, Davies et al.~\cite{davies2020tale} improved it to a $(4 + \epsilon)$-approximation. More recently, Chiarelli et al.~\cite{chiarelli2022fair} considered the Santa Claus problem assuming conflicting items represented by a conflict graph. They analyzed the NP-hardness of the problem for specific subclasses of conflict graphs and provided pseudo-polynomial solutions for others. Our work complements their results by providing constant approximate (polynomial time) solutions for interval graphs with uniform and binary valuations for course allocation.

Another well-studied fairness criterion in the fair division literature is envy-freeness~\cite{fisher2002resource}, where every agent values her allocation at least as much as she values any other agent’s
allocation. However, envy-freeness does not translate well when the items to be allocated are indivisible (for example, if there is one indivisible item and two students, the item can be allocated to only one student, and the other student would envy). Thus, for indivisible items (such as course seats), an appropriate fairness criterion is envy-freeness up to one item (EF1), defined by Budish~\cite{budish2011combinatorial}. Prior works have shown that an EF1 allocation always exists while allocating non-conflicting budgeted courses~\cite{budish2011combinatorial}, under submodular valuations~\cite{lipton2004approximately}, under cardinality constraints~\cite{biswas2018fair},  conflicting courses with monotone submodular valuations and binary marginal gains over the courses~\cite{benabbou2021finding,viswanathan2022yankee}, and many more.  However, these results do not consider interval graphs to model conflicting courses and thus, the existing EF1 solutions cannot solve the fair course allocation problem that we consider. Recent work by Hummel et al.~\cite{hummel2022fair} explored the allocation of conflicting items with EF1 fairness criteria. They showed the existence of EF1 for conflict graphs with small components and refuted the existence of EF1 when the maximum degree of the conflict graph is at least as much as the number of agents. Moreover, they provided a polynomial time EF1 solution when the conflict graph consists of disjoint paths and the valuations are binary. Our work extends their results by providing a polynomial time EF1 solution for interval graphs with binary valuations, which are more general than disjoint-path graphs and capture conflicts between courses. 

Fair allocation of intervals has been studied in job scheduling problems, where each job is represented as an interval (with a starting time, deadline, and processing time) and is required to be allocated to machines such that the same machine is not scheduled to run another job at the same time. Fairness notions considered are in terms of load balancing~\cite{ajtai1998fairness}, waiting time envy freeness~\cite{bilo2014price}, completion time balancing~\cite{im2020fair}, and EF1 among machines~\cite{minmingli}. However, these papers allow flexible time intervals, which cannot capture conflicts as graph edges and represent a different problem from our work. 

Other related techniques to our fair course allocation problem include equitable coloring~\cite{bodlaender2005equitable,kierstead2010fast,GOMES202159}, bounded max coloring~\cite{bonomo2011bounded,hansen1993bounded},
mutual exclusion scheduling~\cite{gardi2008mutual,jansen2003mutual,marecek2019semidefinite}, although most of these works are only tangentially related to our problem at hand. There have also been many works 
on approximation algorithms for various different types of conflict models~\cite{bodlaender1994scheduling,cai2018approximation,kowalczyk2017exact,mallek2022scheduling} and resource constrained scheduling~\cite{bendraouche2015scheduling} but none of these works operate in the specific
conflict graph and allocation model
studied in our paper.

\subsection{Summary of Contributions}
In this paper, we tackle the problem of \textit{fair} allocation of conflicting resources. We prove that a general version of the problem is NP-hard via a reduction from the independent set problem in~\cref{b-match}. This motivates the study of a specific class of conflict graphs, namely interval graphs, which capture the course allocation problem. For interval graphs, we provide polynomial time algorithms to obtain a fair allocation. We establish that, oftentimes, \textit{simple algorithms are enough to provide multiple guarantees in terms of efficiency and fairness}, specifically, a round robin approach is often sufficient. Figure~\ref{fig:overview-of-results} summarizes our results. Our main results are: 

\begin{itemize}
    \item We first consider \textit{uniform utilities} in \cref{sec:uniform} and show that the course allocation under the time conflicts problem with the objective of maximizing social welfare is NP-complete in general. However, we develop polynomial-time solutions when there are a constant number of students \emph{or} when the credit caps and course lengths are uniform. We further provide solutions that have fairness guarantees, one of which satisfies envy-freeness up to any good (EFX) and the other achieves approximate maxi-min fairness. 
    \item We then investigate \textit{binary utilities and uniform credits for all courses} in \cref{sec:bin} and develop a $(1/2)$-approximate solution for the course allocation problem under the time conflicts problem with the objective of maximizing social welfare. We further provide solutions that have fairness guarantees, one of which satisfies envy-freeness up to one good (EF1) and the other achieves approximate max-min fairness. 
   
    \item Our experimental evaluation demonstrates that our algorithms yield near-optimal solutions on synthetic as well as real-world university datasets.
\end{itemize}

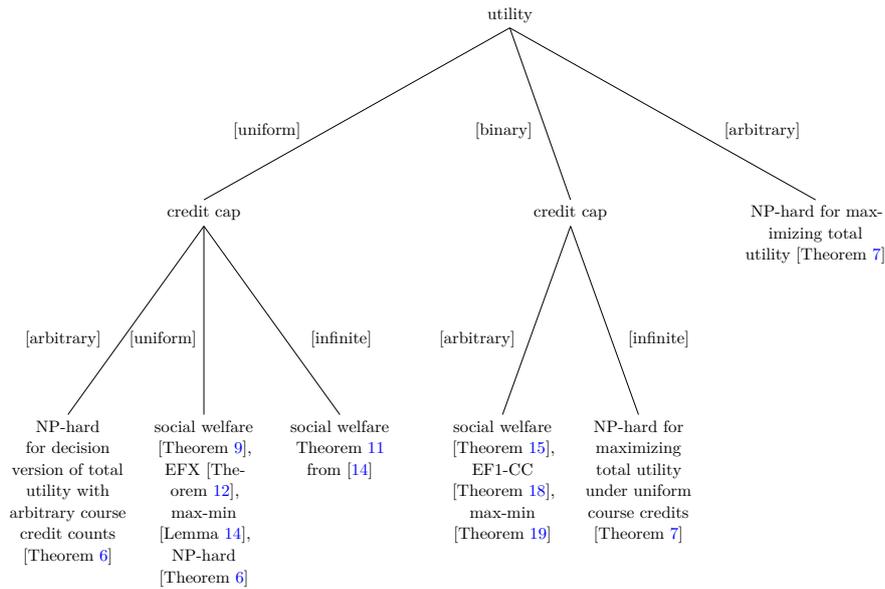
\begin{figure}[!htbp]
\centering
\resizebox{12cm}{!}{
\begin{tikzpicture}
\tikzset{text width = 70pt, align=center}
\tikzset{level 1/.style={level distance=120pt, sibling distance = 5pt}}
\tikzset{level 2/.style={level distance=130pt, sibling distance = 5pt}}
\tikzset{level 3/.style={level distance=120pt, sibling distance=30pt}}
\tikzset{every level 1 node/.style={text width=95pt}}
\tikzset{every level 2 node/.style={text width=70pt}}
\tikzset{every level 3 node/.style={text width=100pt}}

\Tree
[.utility
  \edge node[pos=0.6, left, xshift=3pt] {[uniform]};
  [.{credit cap}
    \edge node[pos=0.6, left, xshift=3pt] {[arbitrary]};
    [.{NP-hard for decision version of total utility with arbitrary course credit counts [\cref{thm:uni-utility-arb-cred-cap-hardness}}] ]
    \edge node[left, pos=0.6, xshift=14pt] {[uniform]};
    [.{social welfare [\cref{thm:unif-creds-lengths-utilities}],  EFX [\cref{thm:EFX-unif-cred-caps-utilities}], max-min [\cref{lem:max-min-guarantee-2}], NP-hard [\cref{thm:uni-utility-arb-cred-cap-hardness}]} ]
    \edge node[pos=0.6, right, xshift=-5pt] {[infinite]};
    [.{social welfare \cref{thm:no-cred-caps-unif-utilities} from \cite{CARLISLE1995225}} ]
  ]
  \edge node[pos=0.6, left, xshift=13pt] {[binary]};
  [.{credit cap}
    \edge node[pos=0.6, left, xshift=7pt] {[arbitrary]};
    [.{social welfare [\cref{thm:1/2-binary-utility-arb-credit-cap-credits}], EF1-CC [\cref{thm:round-robin-augmented-ef1}], max-min
    [\cref{thm:max-min-binary}]} ]
    \edge node[pos=0.6, right, xshift=-10pt] {[infinite]};
    [.{NP-hard for maximizing total utility under uniform course credits [\cref{thm:CAUTC-hardness}]} ]
  ]
  \edge node[pos=0.6, right, xshift=3pt] {[arbitrary]};
  [.{NP-hard for maximizing total utility [\cref{thm:CAUTC-hardness}]} ] 
  ]
\end{tikzpicture}
}
\caption{Overview of results.}
\label{fig:overview-of-results}

\end{figure}

\section{Preliminaries}\label{sec:prelims}
In this section, we define our problem as well as the necessary concepts for our results. We first define 
our main problem which we call 
the Course Allocation Under Time Conflicts problem (CAUTC). This problem describes an issue almost all universities face: given a set of courses that have meeting times during the week and student preferences over these courses, what is the best way to assign these courses to students? Each course has a seating capacity, after all. From a university's perspective, filling seats has value (maximizing utility), but we have to balance that with a fairness aspect as well.

\subsection{Course Allocation under Time Conflicts Model}\label{sec:model}
We consider the problem of allocating a set of $m$ courses among a set of $n$ students. Let $\mathcal{N}$ be the set of students and $\mathcal{M}$ be the set of courses. Courses in $\mathcal{M}$ have indices in $\mathcal{M}$. Each student $i\in\mathcal{N}$ has a non-negative utility for each course $j\in\mathcal{M}$; this utility is denoted by $u_{i}(j)\geq 0$. $C_i$ represents the maximum number of credits, a student $i$ can take. Each course $j$ has a certain number of credits indicated by $c_j$, a seat capacity of $s_j$ for each $j \in \mathcal{M}$, a start and end time, represented by the tuple $(\starttime_j, \etime_j)$ 
and a duration $d_j$ (in units consisting of discrete time steps). Finally, each course $j$ 
is associated with a seat count $s_j$.
Therefore, the restrictions are:

\begin{itemize}
    \item A student $i \in \mathcal{N}$ can be matched to courses with the total credits at most $C_i$ (\textit{credit cap}).
    \item A course $j \in \mathcal{M}$ can be allocated to at most $s_j$ students.
    \item No student can be allocated a pair of courses that overlap in time.
\end{itemize}

Although we define the problem in the most general form, for the rest of this paper, we set $c_j = 1$ for all courses.
Furthermore, we reduce to the equivalent problem where we make a copy of the course for each seat and create an interval
with the same start and end time for each seat of the course. Via this reduction, we also set $s_j = 1$ for all courses.

The course schedule can be represented as an interval graph.
We illustrate such a configuration in~\cref{fig:course-interval-graph}.\\ 

\begin{figure}[htbp]
  \centering
  \includegraphics[width=0.9\textwidth]{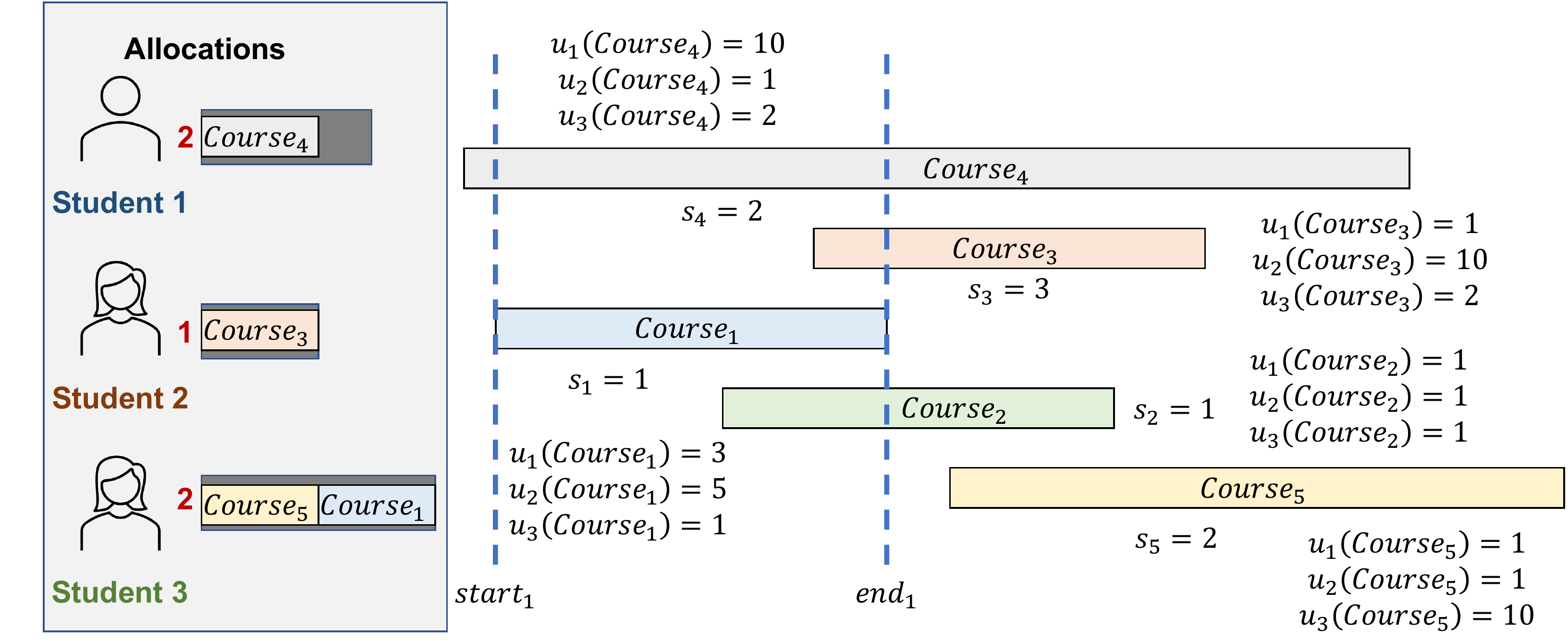}
  \caption{An CAUTC instance with $3$ students and $5$ courses, with one seat per course. All courses conflict with each other except for Course$_1$ and Course$_5$. The red numbers students indicate the credit caps for students. The allocation represents a solution for CAUTC-SW (\cref{def:cautcsw}).}\label{fig:course-interval-graph}
\end{figure}

\subsection{Fairness Measures}
We first consider the problem of finding an allocation that maximizes the social welfare (total sum of utilities of all the students based on the courses allocated) subject to all the feasibility and non-conflicting constraints. We call this maximization problem \CAUTCSW.


\begin{definition}[\CAUTCSW]\label{def:cautcsw}
Given a set of students $\mathcal{N}$, a set of courses $\mathcal{M}$, and the set of utility functions $\mathcal{U}$, \CAUTCSW is the assignment of courses to students such that the social welfare is maximized and
the constraints of CAUTC are satisfied.
\end{definition}

In addition to maximizing social welfare, we also consider a number of common fairness measures as constraints. We first define them here but will slightly modify some of these definitions in their respective sections later on in
this paper.

We first define the concept of \emph{envy-free up to any
good} (EFX). Informally, EFX means that if any agent A were to be envious of any agent B, then A would no longer be envious if any one item were to be removed from agent B's allocation. 

\begin{definition}[Envy-Free Up to Any Good (EFX)]\label{def:any-good}
For all students $i \in \mathcal{N}$, if there exists an $i' \in \mathcal{N}$ such that $u_i(A_{i'}) > u_i(A_i)$, then for all items $x \in A_{i'}$, it follows that $u_i(A_{i'} \setminus {x}) \leq u_i(A_i)$ or $C_i = \sum_{j \in A_i} c_j$ (student $i$ has reached their credit cap), where $A_k$ denotes the allocation of courses to student $k$.
\end{definition}

A slightly weaker version of EFX is \emph{envy-free up to one
good} (EF1), defined below. Informally, EF1 means that if any agent A were to be envious of any agent B, then A would no longer be envious if a particular item were to be removed from agent B's allocation. 

\begin{definition}[Envy-Free Up to One Good (EF1)]\label{def:one-good}
For all students $i \in \mathcal{N}$, if there exists an $i' \in \mathcal{N}$ such that $u_i(A_{i'}) > u_i(A_i)$, then there exists an item $a \in A_{i'} $ satisfying $ u_i(a)>0$, such that $u_i(A_{i'} \setminus {a}) \leq u_i(A_i)$
or $C_i = \sum_{j \in A_i} c_j$ (student $i$ has reached their credit cap), where $A_k$ denotes the allocation of courses to student $k$.
\end{definition}

The problem with only ensuring EF1 is that there is a trivial allocation of courses consisting of giving everyone one course only or no courses. Such an allocation is EF1 since no one envies anyone else by more than one course. However, such an allocation is not a very useful allocation most students would not receive as many courses as they want and there will be many remaining courses. Thus, we need a better measure of envy. A definition from \cite{minmingli} resolves this problem. Suppose all unassigned courses in each iteration were donated to a dummy student, the \defn{charity}, who is unable to envy anyone, but students are able to envy the charity. Then, having the charity resolves the issue of trivial solutions. Specifically, any student $i$ can envy the charity by considering the maximum independent set among the courses in the charity that are desired by $i$. If such a maximum independent set is larger than the number of courses allocated to $i$, then $i$ envies the charity. We formally define EF1 Considering Charity (EF1-CC) to be our new notion of envy below. 

\begin{definition}[Envy-Free Up to One Good Considering Charity (EF1-CC)]\label{ef1-cc}
Any student $i$ who has reached their credit cap (i.e. $C_i = \sum_{j \in A_i} c_j$) does not 
envy anyone else. 
For all other students $i, i' \in \mathcal{N}$ (who have 
not reached their credit caps) and given an 
allocation $\mathcal{A} = (A_1, \ldots, A_i, \ldots A_n)$ of courses, it holds that
$|\{j \mid u_i(j) > 0, j \in A_i\}| \geq |\{j \mid u_i(j) > 0, j \in A_{i'}\}|-1$. Let $D$ be the 
set of courses that are unassigned and held by a 
dummy student defined as the \defn{charity}. 
Let $MIS_i = MIS(\{j \mid u_i(j) > 0, j \in D\})$ be the 
maximum independent set of courses in $D$ that are
desired by student $i$. Then, for all students $i \in \mathcal{N}$, it holds that $|\{j \mid u_i(j) > 0, j \in A_i\}| \geq |MIS_i| - 1$.
\end{definition}

Finally, we consider a Santa Claus fairness objective which is
to maximize the minimum allocation of courses to any
student. For simplicity, we denote this problem as {CAUTC-SC}. 

\begin{definition}[CAUTC-SC]\label{cautc-sc}
    Determine an allocation of courses to students 
    $\mathcal{A} = (A_1, \ldots, A_n)$ that maximizes
    the minimum utility of any student subject to the constraints of CAUTC. Namely, we
    seek to satisfy the following objective
    $\max_{\mathcal{A}}\left(\min_{i \in \mathcal{N}}
        \left(\sum_{j \in A_i} u_i(j)\right)\right)$.
\end{definition}

\section{Uniform Utilities for Courses}\label{sec:uniform}

In this section, we discuss the setting where all 
students have equal, uniform preferences for all 
courses. In other words, in this section, all students
have preference $1$ for every course. In this setting,
we show a number of hardness, social welfare,
and fairness results described in the following 
sections. 

\subsection{Hardness of \CAUTCSW under Uniform Utilities}

We show that \CAUTCSW is NP-hard (\cref{thm:uni-utility-arb-cred-cap-hardness}). Subsequently, we consider some variants of the problem that are polynomial-time solvable in the following sections. 

\begin{theorem}\label{thm:uni-utility-arb-cred-cap-hardness}
The \CAUTCSW problem where the utilities are uniform, credit caps are uniform, course are non-overlapping, and number of 
credits for each course is non-uniform and arbitrary is NP-hard.
\end{theorem}

We prove this via a reduction from the 3-partition problem (~\cref{app:uni-utility-arb-cred-cap-hardness}).

\begin{theorem}\label{thm:CAUTC-hardness}
The \CAUTCSW problem where utilities are binary, credit caps are infinite, and number of credits for each course is uniform is NP-hard.
\end{theorem}

We prove this via a reduction from the $k$-coloring problem for circular-arc graphs. The complete proof is in~\cref{app:CAUTC-hardness}.

\subsection{Maximizing Social Welfare}
In this section, we show that, for some more 
restricted settings, the \CAUTCSW problems are 
polynomial-time solvable. We first show that when
given a constant number of students, we can efficiently solve
the most general form of the problem with no
restrictions on either the credit caps or the number of credits for each course, and with arbitrary preferences for each student.

\begin{algorithm}[H]
\caption{Round Robin Algorithm for \CAUTCSW}\label{alg:round-robin}
\begin{algorithmic}[1]
\Require{Set of students $\mathcal{N}$, set of courses $\mathcal{M}$, uniform (unit) utilities} 
\Ensure{Assignment of courses to students.}
\Function{RoundRobin}{$\mathcal{N}$, $\mathcal{M}$}
\State Sort $\mathcal{M}$ chronologically by earliest finish time.\label{line:rr-finish-time}
\State Initialize student assignments $\mathcal{A}$ to empty sets. \Comment{each student starts out with no courses}
\For{course $j \in \mathcal{M}$ in sorted order}\label{line:rr-iterate}
        \State Let $T = \left\{s \mid |A_s| < C_s, \text{no course in $A_s$ conflicts with $j$}\right\}$.\label{line:rr-set-t}
        \If{$|T| > 0$}
            \State Let $s = \min_{s' \in T} \left(|A_{s'}|\right)$ (breaking ties by student index).\label{line:rr-min-classes}
            \State Update $A_s = A_s \cup \{j\}$ \Comment{Assign course $j$ to student $s$} \label{line:rr-assign-class}
        \EndIf
\EndFor
\State \Return $\mathcal{A}$
\EndFunction
\end{algorithmic}
\end{algorithm}

\begin{theorem}\label{thm:constant-students-DP}
\CAUTCSW is polynomial-time solvable when there are only a constant number of students and credit counts for courses
can be distinct but are each $O(1)$. 
\end{theorem}

The proof of~\cref{thm:constant-students-DP} can be found in~\cref{constant-students-DP}.

\begin{theorem}\label{thm:unif-creds-lengths-utilities}
\cref{alg:round-robin} solves \CAUTCSW in $O((n+m)\log n)$ time when there are (1)~uniform credits for all courses, i.e.\ $c_j = c_{j'}$ for all $j, j' \in \mathcal{M}$, (2)~uniform course lengths, i.e., $d_j = d_{j'}$ for all $j, j' \in \mathcal{M}$, and (3)~uniform utilities i.e., $u_i(j)=u_{i'}(j)$ for 
all $i, i' \in \mathcal{N}$. 
\end{theorem}

We prove~\cref{thm:unif-creds-lengths-utilities} via a variation of the greedy-comes-first strategy; we
present our full proof in~\cref{sec:greedy-comes-first}. When the durations of the courses are not uniform,
we can obtain a $(1/2)$-approximate allocation for \CAUTCSW.

\begin{lemma}\label{lem:max-min-guarantee}
There is a $O((n + m)\log n)$ time round-robin
algorithm for \CAUTCSW that obtains a $1/2$-approximation when there are (1)~$n$ students, (2)~uniform credit caps i.e. for any pair of students $i, i' \in \mathcal{N}$, we have $C_i = C_{i'}$, and (3)~uniform utilities i.e. for any pair of students $i, i' \in \mathcal{N}$ and jobs $j, j' \in \mathcal{M}$, we have $u_i(j)=u_{i'}(j')$. 
\end{lemma}

\begin{proof}
    We use the same algorithm as before, given in~\cref{alg:round-robin}. 
    However, we use a slightly 
    different analysis which is somewhat 
    more complicated than our utility proof before but 
    with the same essential flavor of proof using $D_i,
    J_i, B_i$. 
    Namely, the one additional property we prove is that 
    when $|B_i| + |D_i| \geq |J_i|$, our new greedy
    algorithm will pick $|J_i|$ instead of $B_i \cup D_i$. Suppose for contradiction that $i$ picked
    $B_i \cup D_i$ instead of $J_i$, then $i$ must 
    have picked a course with \emph{earlier or the 
    same end time} as each of the courses in $J_i$. 
    We now show that $|B_i \cup D_i| \geq |J_i|$. We prove this through the classic greedy stays ahead proof technique. If one were to chronologically order $B_i \cup D_i$ by finish time and also chronologically order $J_i$ by finish time, and call the two ordered sets as $P$ and $Q$, respectively, and let $P_i$ denote the $i$-th course in set $P$;
    we will prove that it is always true that for all indices $i \leq |J|$, $f(P_i) \leq f(Q_i)$, where $f(x)$ means the finish time of course $x$. Also define the start time function of course $x$ as $s(x)$. The base case of $i=1$ is obviously true due to the nature of the algorithm. Now for the inductive case, assume inductive hypothesis $f(P_i) \leq f(Q_i)$ and we want to prove $f(P_{i+1}) \leq f(Q_{i+1})$. We know that $f(Q_i) \leq s(Q_{i+1})$. Combining this with the inductive hypothesis, we get $f(P_i) \leq s(Q_{i+1})$, so $Q_{i+1}$ is available for our algorithm to choose, and since our algorithm chooses an available course with the earliest end time, $f(P_{i+1}) \leq f(Q_{i+1})$. 
    
    Let's assume for the sake of contradiction that $|J| > |B_i| \cup |D_i|$. Through the same argument as in the inductive case above, say $|B_i| \cup |D_i| = p$, then the start time of $Q_{p+1}$ must have a start time later than the finish time of the last course in $|B_i| \cup |D_i|$, i.e. $s(Q_{p+1}) \geq f(P_p)$, but that means our algorithm would have selected $Q_{p+1}$ (some time) after selecting $P_p$, a contradiction.
\end{proof}

For completeness, we state the following form 
formulation of~\CAUTCSW that is solved via an interval coloring algorithm of Carlisle and Lloyd \cite{CARLISLE1995225}.

\begin{theorem}[\cite{CARLISLE1995225}]\label{thm:no-cred-caps-unif-utilities}
\CAUTCSW can be solved in polynomial time when there are (1)~$n$ students, (2)~no credit caps i.e., $C_i=m$, and (3)~uniform utilities i.e. for any pair of students $i, i' \in \mathcal{N}$, we have $\utility{i}{j}=\utility{i'}{j}$. 
\end{theorem}

\subsection{Guaranteeing Envy-Freeness Up to Any Good}\label{sec:fair}
Maximizing seat occupancy is a reasonable objective only from a financial perspective for the university, but oftentimes, maximizing seat occupancy could result in highly unfair schedules for the students. For example, student A might get all of his favorite courses while student B gets none of his desired courses. We, therefore, consider \CAUTCSW under several fairness notions, such as envy-free up to any good (\cref{def:any-good}) and envy-free up to one good (\cref{def:one-good}).



\begin{theorem}\label{thm:EFX-unif-cred-caps-utilities}
There is an $O((n+m)\log n)$-time algorithm for \CAUTCSW that is EFX when there are (1)~$n$ students, (2)~uniform credit caps i.e. for any pair of students $i, i' \in \mathcal{N}$, we have $C_i = C_{i'}$, and (3)~uniform utilities i.e.\ for any pair of students $i, i' \in \mathcal{N}$ and any pair of jobs $j, j' \in \mathcal{M}$, we have $u_i(j)=u_{i'}(j')$. 
\end{theorem}
\begin{proof}
Our algorithm is the same round robin algorithm given in~\cref{alg:round-robin}. We first prove the following 
lemma.

\begin{theorem}\label{lem:no-more-assignment}
When student $i$ is no longer able to choose a feasible course, there will be at most $n-1$ courses that can be assigned after $i$'s turn and each of these courses is assigned to a different student.
\end{theorem}
\begin{proof}
Because utilities are uniform, if student $i$ is no longer able to choose a course, this means that all remaining courses
conflict with the courses they are assigned. Suppose the last course that is assigned to student $i$ is course $j$. 
Because we are assigning courses in~\cref{alg:round-robin} in a round robin 
manner in an order determined by non-decreasing end time, all remaining courses (yet to be considered by the 
algorithm) that can be assigned have end time
no earlier than the end time of $j$. Let this set of courses be $A$. 
Since $i$ is no longer able to receive a course, either there remains only $n-1$
courses or $A$ has at least $n-1$ courses and at least $|A| - n + 1$ courses
in $A$ all conflict with $j$. Since all courses in $A$ have end time no earlier than the 
end time of $j$, these $|A| - n + 1$ courses all conflict with each other. In either of these two cases, at most 
$n-1$ courses can be assigned after $i$'s turn. Furthermore, these courses are assigned to different students. If 
there are at most $n-1$ courses in $A$, then by nature of the algorithm, these courses all have end times later than 
the end times of courses assigned to students; furthermore, the ending time of the last course assigned to 
each student can be no later than the end time of $j$ by the nature of our algorithm. Hence, two such courses
can be assigned to one student, then one of these courses can be assigned to $j$. Thus, since we are assigning courses
to a student with the fewest number of courses, each of these courses is assigned to a different student. 
Finally, all additional $|A|- n + 1$ courses all conflict with each other and hence no two of these courses
can be assigned to the same student. 
\end{proof}

Hence, by the time the algorithm completes and by~\cref{lem:no-more-assignment}, 
the cardinalities of all students' allocations are within one of each other, therefore achieving EFX.
\end{proof}

\subsection{Maximizing Max-Min Objective}
In this section, we consider the max-min objective, Santa Claus (SC) problem (\cref{cautc-sc}). We first show that
our algorithm in~\cref{alg:round-robin} gives a
$(1/4)$-approximate CAUTC-SC allocation. 
Specifically, we prove the following. 

\begin{lemma}\label{lem:max-min-guarantee-2}
There is a $O((n + m)\log n)$ time round robin
algorithm (\cref{alg:round-robin}) for CAUTC-SC that obtains a $(1/4)$-approximation when there are (1)~$n$ students, (2)~uniform credit caps i.e. for any pair of students $i, i' \in \mathcal{N}$, we have $C_i = C_{i'}$, and (3)~uniform utilities i.e.\ for any pair of students $i, i' \in \mathcal{N}$ and jobs $j, j' \in \mathcal{M}$, we have $u_i(j)=u_{i'}(j')$. 
\end{lemma}

\begin{proof}
    Given a set of courses with total utility $U$, 
    the max-min value of any allocation is at most 
    $\floor{\frac{U}{n}}$. 
    We now consider two possible cases with 
    respect to the values of $\floor{\frac{U}{n}}$.
    First, we consider the case when 
    $\floor{\frac{U}{n}} \geq 2$. In this case,
    by~\cref{lem:no-more-assignment},
    the max-min value of our allocation is at least
    $\frac{U}{2n} - 1 \geq \frac{U}{4n}$.
    Now, we consider the case when 
    $\floor{\frac{U}{n}} < 2$. In this 
    case, either the max-min value is $0$ or 
    the max-min value is $1$. If the max-min 
    value is $0$, then we trivially obtain our 
    approximation since any allocation will 
    result in the correct approximation. Otherwise,
    if the max-min value is $1$, then there is one
    student who gets only one course. We show
    that if the max-min value is $1$, then our 
    algorithm also allocates at least one course to 
    every student. The criteria for our algorithm 
    giving one course to each student is that there 
    exists at least $n$ courses. Since our
    algorithm assigns the courses in a round
    robin manner, if there are at least $n$ 
    courses, then our algorithm will assign at 
    least one course to each student. In order for 
    the max-min value to be $1$, there must exist
    at least $n$ courses; hence, the max-min value 
    of allocations given by our algorithm matches
    that of the value given in $\opt$.
    Thus, by the two cases we just showed,
    the approximation factor
    is at least $\frac{\frac{U}{4n}}{\frac{U}{n}} 
    = \frac{1}{4}$. 
    \end{proof}

\section{Binary Preferences for Classes with Uniform Credits}\label{sec:bin}

In this section, we discuss the setting where students 
have binary preferences for courses. This is a very 
realistic setting since it is often the case that 
students want to take certain courses and not others. 
We denote the binary preferences of the students  as $\valset: \mathcal{N} \times \mathcal{M} \mapsto \{0,1\} $, where $u_i(c)=1$ denotes that the student $i \in \mathcal{N}$ wants to take the course $c$, and $u_i(c)=0$ denotes that course $c$ is not desired by student $i$. 
If a student has $u_i(c) = 1$, then we say that student
$i$ \emph{desires} course $c$; otherwise, we say that 
student $i$ does not desire course $c$.
Each student $i$ has a credit cap denoted by $C_i$. 
In this section, all courses have uniform number of 
credits; i.e. all courses have the same number of credits. Because of this assumption, we can assume all
courses are $1$ credit each and we scale the credit
caps of each student to the maximum number of 
courses that can fit in the student's schedule.

\subsection{Maximizing Social Welfare}\label{sec:binary-sw}

We first present an algorithm that gives an approximation for \CAUTCSW given binary preferences.
Our algorithm proceeds as follows. Sort the students
by credit cap from largest credit cap to smallest (Line~\ref{line:credit-cap}). 
Then, we iterate the following procedure. 
Let the current student be the first student in the
sorted order of the students by credit cap with no assigned 
courses (Line~\ref{line:mis-iterate}). We find an independent 
set of maximum size among all courses with non-zero
utility for the current student (Line~\ref{line:max-independent-set}). For each 
independent set $I$ and the associated student $i \in \mathcal{N}$, 
we sort the courses in $I$ and give the first
$\max(|I|, c_i)$ courses in $I$ in the sorted order to
student $i$ (Lines~\ref{line:sort-by-end-time},~\ref{line:mis-prefix},~\ref{line:allocate-mis}).
Finally, we remove the allocated courses from the set of available courses (Line~\ref{line:mis-assign-class}).

\begin{algorithm}[H]
\caption{Binary Utilities Algorithm for \CAUTCSW}\label{alg:mis}
\begin{algorithmic}[1]
\Require{Set of students $\mathcal{N}$, set of courses $\mathcal{M}$, binary utilities $U$} 
\Ensure{Assignment of courses to students.}
\Function{MaxIndependentSetRoundRobin}{$\mathcal{N}$, $\mathcal{M}, U$}
\State Sort $\mathcal{N}$ in non-increasing order by credit cap.\label{line:credit-cap}
\State Initialize student assignments $\mathcal{A}$ to empty sets. \Comment{student starts out with no courses}
\For{student $i \in \mathcal{N}$ in sorted order}\label{line:mis-iterate}
        \State Let $I = MIS(\{j \mid j \in \mathcal{M}, u_i(j) > 0\})$. \Comment{Find MIS in remaining courses.}\label{line:max-independent-set}
        \If{$|I| > C_i$}
            \State Sort $I$ by end time.\label{line:sort-by-end-time}
            \State Set $I \leftarrow I[C_i]$. \Comment{Resize the MIS to be the first $C_i$ courses in the MIS.}\label{line:mis-prefix}
        \EndIf
        \State Set $A_s \leftarrow I$.\label{line:allocate-mis}
        \State Update $\mathcal{M} = \mathcal{M} \setminus I$. \Comment{Remove assigned courses.} \label{line:mis-assign-class}
\EndFor
\State \Return $\mathcal{A}$
\EndFunction
\end{algorithmic}
\end{algorithm}

\begin{theorem}\label{thm:1/2-binary-utility-arb-credit-cap-credits}
\cref{alg:mis} solves \CAUTCSW in $O(n^2)$ time with an 
$(1/2)$-approximation when there are $n$
students, arbitrary credit caps $C_i$ for 
all $i \in \mathcal{N}$, unit credits per course
$c_j=1$  for all $j \in \mathcal{M}$, and binary utilities for all 
students, i.e.\ $\utility{i}{j} \in \{0, 1\}$ for all $i \in \mathcal{N}$.
\end{theorem} 

\begin{proof}
In the sorted order of courses by end time in $I$,
if course $j \in \mathcal{M}$ is assigned in $\opt$ and by our 
algorithm, then we skip this course in our analysis. However,
if the course is assigned in $\opt$ but not assigned by our 
algorithm, then we need to argue that either another course
is assigned in its place or that we can \emph{charge} 
it to another assigned course. 
For all of the below 
cases, suppose that course $j \in \mathcal{M}$ is assigned
to student $i \in \mathcal{N}$ in $\opt$ but not assigned in 
our assignment. For simplicity, we denote the assignment 
produced by our algorithm as $\mathcal{A}$. Let $D_i$ be the set
of courses assigned to student $i$ in $\mathcal{A}$ which were not
assigned to any student in $\opt$; let $B_i$ be the set of courses assigned
to $i$ in $\mathcal{A}$ but assigned to $q \neq i \in \mathcal{N}$ in $\opt$.
Finally, let $J_i$
be the set of courses assigned to $i$ in $\opt$ but assigned
to no student in $\mathcal{A}$.
We consider all possible cases below. 

\begin{itemize}
    \item If $|D_i| \geq |J_i|$, then for each course in $J_i$,
    we can replace it with a course in $D_i$ and achieve
    the same maximum total utility.
    \item If $|D_i| < |J_i|$, then we consider two additional 
    cases:
    \begin{itemize}
        \item It is impossible to have $|B_i| + |D_i| < |J_i|$ 
        since $|J_i|$ is a larger independent set 
        and would have been assigned to $i$ instead of 
        $B_i \cup D_i$.
        \item Then, the remaining case is 
        that $|B_i| + |D_i| \geq |J_i|$. This case is the core
        of our proof. In this case, we know that $|B_i| \geq |J_i| - |D_i|$. We pick 
        an arbitrary set of $|D_i|$ jobs in $J_i$
        and replace them each with a unique 
        job in $D_i$. This does not change the 
        optimum total utility value.
        Now, we charge each of the remaining
        $|J_i| - |D_i|$ jobs in $J_i$ to a job in $|B_i|$. 
        We now count the 
        number of ``charges'' that each course in $|B_i|$ gets. Since $|B_i| \geq |J_i| - |D_i|$ and we do not charge a course
        in $B_i$ with any other course not in $J_i$, each course in $B_i$ is charged with at most one charge resulting from a course in $J_i$.
    \end{itemize}
\end{itemize}

We now count the number of courses assigned in 
both $\opt$ and $\mathcal{A}$ as well as the number of charges
each course gets. By the cases above, each of these
courses gets at most $1$ charge. Hence, if each charge is
added to the set of allocated courses, the 
utility increases by at most a factor of two. Hence,
our algorithm produces a $(1/2)$-approximation.
\end{proof}

\subsection{Guaranteeing Envy-Freeness Up to One Good}\label{sec:binary}

Given an allocation of courses to students $\mathcal{A}=(A_1, \ldots, A_i,\ldots, A_n)$ (where $A_i$ is the set of courses assigned to student $i$), a student $i$ is said to \defn{envy}
student $i'$ if the number of student $i$'s desirable courses in $A_i$ is less than that in $A_{i'}$, that is, $|\{j \mid u_i(j) > 0, j \in A_i\}| < |\{j \mid u_i(j) > 0, j \in A_{i'}\}|$. Similarly, an allocation $A$ is called EF1 when for every pair of students $i, i' \in \mathcal{N}$, the following holds: $|\{j \mid u_i(j) > 0, j \in A_i\}| \geq |\{j \mid u_i(j) > 0, j \in A_{i'}\}|-1$. Note that in the binary valuation setting, EF1 implies that, removing
\emph{any} course that $i$ desires from $A_{i'}$ results in
$i$ no longer envying $i'$.
We provide an algorithm (\cref{alg:augmented-ef1}) and prove that this algorithm satisfies the stronger fairness criterion called EF1-CC (\cref{ef1-cc}).

\begin{algorithm}[!ht]
\caption{Round Robin Algorithm for EF1-CC Allocation 
with Binary Utilities}\label{alg:augmented-ef1}
\begin{algorithmic}[1]
\Require{Set of students $\mathcal{N}$, set of courses $\mathcal{M}$, binary utilities $\valset$} 
\Ensure{EF1-CC Allocation for Binary Utilities}
\Function{EF1CCRoundRobin}{$\mathcal{N}$, $\mathcal{M}$, $U$}
\State Sort $\mathcal{M}$ chronologically by earliest finish time.\label{line:finish-time}
\State Initialize student assignments $\mathcal{A}$ to emptysets. \Comment{students start out with no courses}
\For{course $j \in \mathcal{M}$ in sorted order}\label{line:iterate}
        \State Let $T = \left\{s \mid u_s(j) = 1, |A_s| < C_s, \text{no course in $A_s$ conflicts with $j$}\right\}$.\label{line:set-t}
        \If{$|T| > 0$}
            \State Let $s = \min_{s' \in T} \left(|A_{s'}|\right)$ (breaking ties arbitrarily).\label{line:min-classes}
            \State Update $A_s = A_s \cup \{j\}$ \Comment{Assign course $j$ to student $s$} \label{line:assign-class}
        \EndIf
\EndFor
\State \Return $\mathcal{A}$
\EndFunction
\end{algorithmic}
\end{algorithm}

Our algorithm is a simple modification of the round-robin algorithm given in~\cref{alg:round-robin}. The only change we make to the algorithm is that when we perform the round-robin assignment, each course is iteratively assigned to only one of those students who have non-zero utility for the course, in addition to ensuring that the selected student has the minimum number of current courses, has not reached credit cap and has no conflict with the course. Our modified pseudocode is given in~\cref{alg:augmented-ef1}.

Specifically,~\cref{alg:augmented-ef1} first sorts the courses chronologically by finish time (Line~\ref{line:finish-time}).
Then, we iterate over the courses one by one in the sorted order of finish time (Line~\ref{line:iterate}). Among the 
students who have non-zero preference for the course, have not reached their credit caps, and have no conflicts with the course
(Line~\ref{line:set-t}), we select a student (breaking ties arbitrarily) with the least number of assigned courses among
these students (Line~\ref{line:min-classes}). Finally, we assign the course to the student (Line~\ref{line:assign-class}).

\begin{theorem} \label{thm:round-robin-ef1} 
Under binary preferences, uniform credits for all courses, and arbitrary credit caps,
the round-robin algorithm given in~\cref{alg:augmented-ef1} produces an EF1 allocation.
\end{theorem}

\begin{proof}
We prove by induction that for any two students $s$ and $s'$, student $s$ never envies $s'$ by more than one course throughout the entirety of~\cref{alg:augmented-ef1}. The induction is on the finish time of each course in the schedule of $s'$ among the set of courses for which $s$ has non-zero utility, i.e. we induce on the finish times of the set of courses $L = \left[ j \in A_{s'} \mid \utility{s}{j} > 0 \right]$ sorted from earlier to later times. Notice that $L$ is the set of courses assigned to $s'$ that are desired by $s$, as courses assigned to $s'$ not desired by $s$ cannot make $s$ envy $s'$ and therefore irrelevant to this proof. Now, for each $i \in [|L|]$, 
we consider the set of courses
assigned to both $s$ and $s'$ which has end time 
no later than the end time of $L[i]$. For
simplicity, we use the phrase \emph{by the time
course $L[i]$ ends} to mean that we consider 
the set of courses held by $s$ and $s'$ with 
end time no later than $L[i]$.

\begin{lemma} \label{inductive-round-robin-ef1}
For each course $L[i]$ for all $i \in [|L|]$, at the time $L[i]$ ends, student $s$ envies $s'$ by at most one course.
\end{lemma}

\begin{proof}
We prove via induction on the $i$-th course of $L$
which ends at time $e_i$.
The base case is when $i = 1$. Student $s$ trivially envies $s'$ by at most $1$ because if $s$ has no courses by the time course $L[i]$ ends
then $s$ will only envy $s'$ by $1$; otherwise, 
$s$ will not envy $s'$.

We assume for the purposes of induction that
$s$ envies $s'$ by at most one course by the time $L[i]$ ends. We now prove that $s$ envies $s'$ by at most one course by the time $L[i+1]$ ends. By our induction hypothesis, there are two cases, when $s$ envies $s'$ by one course when $L[i]$ ends, and when $s$ does not envy $s'$ when $L[i]$ ends. In the latter case, it is only possible for $s$ to envy $s'$ by at most one course by the time $L[i+1]$ ends since $s'$ has gained at most one 
additional course which $s$ desires by the time 
$L[i+1]$ ends. Now we prove the former case. Let $j$ be the next course (after the course
$L[i]$) that the algorithm 
considers that is assigned to either student $s$ or $s'$, is desired by $s$. 
Then, course $j$ would fit into the current schedule of both $s$ and $s'$, since $j$ 
starts after the end time of $L[i]$. Suppose for the sake of contradiction 
that $j$ is assigned to $s'$. Since we compare the set of courses
that end no later than the end time of $L[i]$, if $j$ is assigned to $s'$ then
$j$ has start time later than $L[i]$. 
Student $s$ envies $s'$ by $1$ course among the set of courses she received
that end no later than $L[i]$. Then, course $j$ is not assigned to $s$ only if
$s$ has a conflicting course (since $s$ has fewer courses than $s'$); 
however, this contradicts with $j$ being the 
next course assigned after $L[i]$ to either $s$ or $s'$.
\end{proof}

Now to prove~\cref{thm:round-robin-ef1}, we use~\cref{inductive-round-robin-ef1}. Specifically, by the time the last course in $L$ ends, student $s$ envies $s'$ by at most one course. Any course in the schedule of $s$ that ends at a time later than this does not increase the envy $s$ feels towards $s'$. And due to symmetry, $s'$ similarly does not envy $s$ by more than one course. Similarly, any course assigned to $s$ in between the ending times of $L[i]$ and $L[i+1]$ does not increase the envy of $s$.
\end{proof}

\begin{theorem}\label{thm:round-robin-augmented-ef1}
Under binary preferences and uniform credits for all courses,~\cref{alg:augmented-ef1} produces an EF1-CC allocation.
\end{theorem}

\begin{proof}
\cref{thm:round-robin-ef1} stated that no student envies another student by more than one course. We are left to show that no student envies charity by more than $1$ course. 
Assume for the sake of contradiction that there is a student $s$ that envies the charity, this means that (1) $|A_s|<C_s$ where $c_s$ is the credit cap for student $s$, and (2) there is a bigger independent set of courses (name this set $I$) 
among the courses assigned to the charity than the number of allocated courses to $s$, i.e.\, $|A_s| < |I|$. 

First, all courses in $I$ overlap with $A_s$ because if some course $j \in I$ does not conflict with any course in 
$A_s$, then our algorithm would have assigned $j$ to $s$. 
If we were to sort $I$ and $A_s$ by earliest finish time first and index them by $i$, observe that for all $i$, course $A_{s}[i]$ ends earlier than $I[i]$ due to our algorithm (this can be proven with a very elementary greedy stays ahead induction proof \cite{KleinbergTardosBook}). This means that $|A_s| \geq |I|$ because if there were to be a course $j = E[|A_s| + 1]$, that means $j$ begins after the last course in $A_s$ ends, which means our algorithm would have assigned $j$ to $s$.
\end{proof}

\subsection{Maximizing Max-Min Objective} \label{sec:max-min}
Now, we look at a more general version of CAUTC-SC considering binary utilities and provide the following algorithm that gives a constant factor 
approximation when the maximum and minimum durations of any course are within a constant factor $c$ of each other.
We first describe our algorithm with the pseudocode
provided in~\cref{alg:quanquan}. The algorithm proceeds
as follows. The courses are sorted by end time (Line~\ref{mm:sort}). Then,
in the sorted order of courses, each course is given
to a student who has non-zero preference for the course,
has not filled up all of their credits (up to their credit cap), has no 
conflicting courses, and who has the least number of assigned 
courses among all students who have non-zero preference
for the course (Line~\ref{mm:min-assignment}). Suppose we assign course $j$ to a student
$i$. Let $d_i$ be a \emph{dummy course} that we create
for each student $i$. Then, we repeatedly 
perform the following 
procedure until no more \emph{augmenting paths} exist (Line~\ref{mm:augmenting-path}):

\begin{itemize}
    \item For each course assigned to student $i$, draw
    a directed edge from course $j'$ assigned to 
    student $i' \neq i$
    if $j$ conflicts with $j'$ and removing $j'$ means that
    $j$ does not conflict with any other course assigned to $i'$
    and $i'$ has less assigned courses than $i$ (Line~\ref{mm:reassign}). 
    \item For each course assigned to student $i$, draw
    a directed edge from dummy course $d_{i'}$ to $j$
    if $j$ does not conflict with any course 
    assigned to $i'$ and $i'$ has less than or equal
    to the number of courses assigned to $i$ (Line~\ref{mm:dummy-edge}).
    \item Repeat with the courses assigned to $i'$
    and omit all courses assigned to student $i$
    from this part of the graph construction.
\end{itemize}

Once a full directed acyclic graph is drawn using the 
above procedure, we define an \emph{augmenting path}
to be a directed path with the source at a dummy course
and sink at a course of $i$ (Line~\ref{mm:aug-path}). We repeatedly produce a new
directed acyclic graph using the above procedure and 
switch courses between students via an augmenting path
until no such augmenting paths remain (Line~\ref{mm:no-aug-path}). Then, we proceed
with assigning the next item in the sorted order of courses. 
We prove that our algorithm returns a constant factor approximation of the max-min objective value. 

\begin{algorithm}[!hbt]
\caption{Max-Min Assignment of Courses}\label{alg:quanquan}
\begin{algorithmic}[1]
\Require{Courses $\mathcal{M}$, students $\mathcal{N}$, binary utilities $\valset$} 
\Ensure{Approximate max-min allocation $J$}
\Function{Find-Max-Min-Allocation}{$\mathcal{M}$, $\mathcal{N}$, $\valset$}
\State Sort courses in $\mathcal{M}$ by end time from earliest
to latest.\label{mm:sort}
\State $D \leftarrow \emptyset$.
\State Let $Q \leftarrow \emptyset$ be a queue of students.
\For{each course $j$ in sorted order}
    \State Assign $j$ to student $i$ with minimum number of 
    assigned courses, has not reached credit cap, where $u_i(j) > 0$, 
    and does not have any conflicting courses.\label{mm:min-assignment}
    \State Add $i$ to the end of $Q$.
    \State Set $AugPath \leftarrow True$. \label{mm:augmenting-path-exists}
    \While{$AugPath$}\label{mm:augmenting-path}
    \While{$Q \neq \emptyset$}
        \State Remove the first student $i'$ from $Q$.
        \For{each course $j$ assigned to $i'$}
            \State Draw directed edge from $j'$ assigned
            to student $b$
            to $j$ if $j'$ conflicts with $j$,
            removing $j'$ results in $j$ conflicting
            with no course assigned to $b$ conflicting
            with $j$ and $b$ now has less assigned courses
            than $i'$, and $b \not\in D$. Add $b$ to the end of $Q$.\label{mm:reassign}
            \State Draw a directed edge from $d_{b}$ to $j$ if student $b$ does not 
            have any courses that conflict with $j$ and $b$ has at most 
            as many courses as $i'$. Add $b$ to the end of $Q$.\label{mm:dummy-edge}
        \EndFor
        \State $D \leftarrow D \cup i'$.
    \EndWhile
    \State Find an augmenting path with source
    at a dummy course and sink at course assigned to
    $i$ and reassign courses along
    augmenting path from sink to source. \label{mm:aug-path}
    \If{there is no augmenting path}
        \State $AugPath \leftarrow False$.\label{mm:no-aug-path}
    \EndIf
    \EndWhile
\EndFor
\State \Return Allocation of courses to students.
\EndFunction
\end{algorithmic}
\end{algorithm}

\begin{theorem}\label{thm:max-min-binary}
    \cref{alg:quanquan} achieves a $c$-factor approximate solution for CAUTC-SC,  where $c$ is the maximum ratio between the durations of any two courses. 
\end{theorem}

\begin{proof}
Let $S$ denote the set of students with the minimum number of assigned courses by our algorithm. We compare the allocations of courses assigned to each of the students in $S$ by our algorithm 
with the allocation of courses assigned to the students by $\opt$. Let $i \in S$ be one such student. 
Let $A_{i}$ be the set of courses allocated to student $i$ by our algorithm and $\opt_{i}$ be the 
set of courses allocated to $i$ by $\opt$. There are four different types of courses assigned to these students
that we are concerned with. Courses assigned to $i$ in $A_i$ and not in $\opt_i$ can only make max-min greater; thus, we do not consider such courses. The same holds for courses assigned in $A_{i}$
and by $\opt$ to another student. Then, courses assigned by $\opt$ but not assigned
to $A_i$ must conflict with at least one other course assigned to $i$. Hence, such 
courses can be charged to the course that it conflicts. The conflicting course(s) cannot be assigned 
in $\opt_i$; thus, the course in $\opt_i$ can be charged to one of the conflicting
courses. The remaining type are courses that are in $\opt_{i}$, not in $A_{i}$, but are instead assigned to
another student by our algorithm. 
Let $j$ be one such course; then, either

\begin{itemize}
    \item Course $j$ is assigned to a student $i'$ with \emph{less} assigned courses than $i$. This scenario is impossible
    by definition of $i$ as a student with the smallest number of assigned courses.
    \item Course $j$ is assigned to a student $i'$ with the same or 
    more assigned courses than $i$. Student $i$ must be assigned a conflicting course to $j$, as otherwise, when the last course
    assigned to $i'$ is 
    assigned to $i'$, course $j$ would have been transferred to $i$.  Suppose first that
    $i'$ has a greater 
    number of courses than $i$ and 
    $i$ has no conflicting 
    course with $j$, then this is a 
    contradiction since $j$ would have been eventually transferred to $i$. 
    Now suppose $i$ has a course that conflicts $j$. If this conflicting course has an earlier end time than $j$, then $j$
    can be charged to the conflicting course. Furthermore, any course can conflict
    with at most $c$ different courses assigned to $i$ in $\opt$ by our assumption of the 
    ratio between the longest class and shortest class. Thus, we 
    charge the course to the conflicting course assigned to $A_i$; at most $c$ such courses
    can be charged to any course in $A_i$.
\end{itemize} 
\end{proof}




\section{Experimental Results}

In this section, we present a case study with data derived from MS students at Northwestern. We compare the performance of our algorithms \efonealgo and \scalgo to those of optimal integer programs (IP) implemented using Gurobi \cite{gurobi} in Python. 
There are two integer programs of note: one to get the max-min value, and one to get the assignment maximizing the total social welfare given the max-min value $T$ such that every student must receive at least $T$ courses. We will henceforth refer to both of these integer programs that produce the optima as $\opt$. We implement \efonealgo\ and \scalgo in Python~\cite{yiduo_impl}.  
In \scalgo, after looping through each course, exchange path operations are initiated. The graphs of exchange paths were implemented in NetworkX\cite{networkx} in Python.  The experiments are conducted on a Dell PowerEdge R740 with 2 x Intel Xeon Gold 6140 2.3GHz 18 core 36 threads processors, 192GB RAM, 
dual 10Gbps and 1Gbps NICs. 

The dataset was obtained through a Google Form sent out to Master's students who wished to take computer science courses. They could select and rank up to five courses. 
Since ordinal preferences are beyond the scope of this paper, we only considered the courses they desire (binary valuations).  




\begin{table}[!htb]
\resizebox{\columnwidth}{!}{
\begin{tabular}{ll|ccc|ccc|}
\multicolumn{2}{c|}{Datasets} & \multicolumn{3}{c|}{max-min} & \multicolumn{3}{c|}{total utility} \\ 
 &  & OPT & \efonealgo & \scalgo & OPT & \efonealgo & \scalgo \\ \hline
{\begin{tabular}[c]{@{}l@{}}real-world data\end{tabular}} & \begin{tabular}[c]{@{}l@{}}dataset\end{tabular} & 1 & 1 & 1 & 744 & 624 & 744 \\ \cline{2-8} 
 & \begin{tabular}[c]{@{}l@{}}alteration 1\end{tabular} & 2 & 1 & 2 & 725 & 623 & 725 \\ \cline{2-8} 
 & \begin{tabular}[c]{@{}l@{}}alteration 2\end{tabular} & 3 & 2 & 3 & 686 & 686 & 686 \\ \cline{2-8} 
 & \begin{tabular}[c]{@{}l@{}}alteration 3\end{tabular} & 2 & 2 & 2 & 760 & 760 & 760 \\ \hline
{\begin{tabular}[c]{@{}l@{}}synthetic data\end{tabular}} & example 1 & 2 & 2 & 2 & 7 & 7 & 6 \\ \cline{2-8} 
 & example 2 & 3 & 1 & 2 & 6 & 3 & 6 \\ \cline{2-8} 
 & example 3 & 4 & 3 & 3 & 8 & 6 & 6 \\ \cline{2-8} 
 & example 4 & 1 & 1 & 1 & 5 & 4 & 5 \\ \cline{2-8} 
 & example 5 & 1 & 1 & 1 & 6 & 5 & 6 \\ \cline{2-8} 
 & example 6 & 4 & 4 & 4 & 12 & 12 & 12 \\ \cline{2-8} 
 & example 7 & 4 & 4 & 4 & 12 & 12 & 12 \\ \cline{2-8}
 & example 8 & 4 & 4 & 4 & 12 & 12 & 12 \\ \hline
\end{tabular}
}
\caption{Comparison of utilities.}
\label{tab:table-of-utilities}
\end{table}


In terms of utility and max-min value, both algorithms incurred similar values as that of \opt. \cref{tab:table-of-utilities} compares the max-min value between \opt, \efonealgo, and \scalgo. For almost all instances listed in \cref{tab:table-of-runtimes}, \efonealgo\ was much faster than \opt. Since our input data is not too large, we could compute an optimal assignment by solving the corresponding IP using Gurobi, which is not scalable in general.

\begin{table}[!htb] 
\resizebox{\columnwidth}{!}{
\begin{tabular}{ll|ccc|c|}
\multicolumn{2}{c|}{Datasets} & \multicolumn{3}{c|}{max-min} & \multicolumn{1}{c|}{total utility} \\ 
 &  & OPT & \efonealgo & \scalgo & OPT \\ \hline
{\begin{tabular}[c]{@{}l@{}}real-world data\end{tabular}} & \begin{tabular}[c]{@{}l@{}}dataset\end{tabular} & 64.3614 & 44.893709 & 789.936638 & 1.883833 \\ \cline{2-6}
 & \begin{tabular}[c]{@{}l@{}}alteration 1\end{tabular} & 59.0851 & 44.675743 & 751.317634 & 1.818061 \\ \cline{2-6}
 & \begin{tabular}[c]{@{}l@{}}alteration 2\end{tabular} & 47.1846 & 38.609633 & 490.054614 & 1.447398\\ \cline{2-6}
 & \begin{tabular}[c]{@{}l@{}}alteration 3\end{tabular} & 59.0924 & 44.469586 & 1214.057334 & 1.625204\\ \hline
{\begin{tabular}[c]{@{}l@{}}synthetic data\end{tabular}} &   example 1 & 3.672 & 0.405452 & 1.505546 & 1.518124\\ \cline{2-6} 
 & example 2 & 0.6444 & 0.177736 & 0.831938 & 0.186128\\ \cline{2-6} 
 & example 3 & 2.3873 & 0.263971 & 1.145906 & 0.552925\\ \cline{2-6} 
 & example 4 & 2.5193 & 0.200845 & 1.052790 & 0.524367\\ \cline{2-6} 
 & example 5 & 3.3109 & 0.247982 & 1.344810 & 0.588912\\ \cline{2-6} 
 & example 6 & 1.1906 & 0.242676 & 0.667803 & 0.232696\\ \cline{2-6} 
 & example 7 & 11609.3298 & 594.707318 & 34986.057784 & 877.306024\\ \cline{2-6}
 & example 8 & 634809.8065 & 31554.700315 & 2278023.952813 & 9592.727642\\ \hline
\end{tabular}
}
\caption{Comparison of runtimes in milliseconds. All runtimes correspond to instances in the corresponding cells in \cref{tab:table-of-utilities}. There is only one column under \emph{total utility} because \efonealgo and \scalgo are executed only once, as opposed to the two different linear programs of OPT.}
\label{tab:table-of-runtimes}
\end{table}

The results of our experiments demonstrate the effectiveness of our algorithms. \efonealgo was able to give near-optimal solutions with a significantly reduced computational cost compared to integer programming, a traditional method. The reduced runtime is a testament to the effectiveness of the algorithms and their potential for practical implementation. The findings of this study highlight the potential for further improvement and optimization of these algorithms (especially \scalgo since its runtime has much room for improvement) making them an attractive option for real-world applications. Although \scalgo is slower
than \opt for some of the tested instances, we believe it will be much faster and more scalable on instances larger than what we 
tested in our experiments.

\section{Conclusions and Future Work}
 We investigated the problem of allocating conflicting resources across $n$ agents, while taking into account both fairness as well as overall utility of the assignment. While resource allocation is extremely well studied, cases when the resources have conflicts have not been well studied from an algorithmic perspective. 

Several generalizations of the course allocation problem open up interesting new directions for the fair allocation literature, such as generalizing utilities beyond additive binary and considering non-uniform credits for different courses. Further, each course may be a corresponding collection of time intervals (instead of a single interval). While we assume that our courses meet once a week, this may not be true for the general case where courses might meet on Tues-Thurs or Mon-Wed or Mon-Wed-Fri. If two courses overlap in any of the time windows then there is an edge in the conflict graph between them. However, such considerations would make the problem more challenging since the corresponding conflict graphs would be more complicated than interval graphs.

Going ahead, there are several directions for future research that can extend and improve upon our approach to course allocation. By addressing these challenges, we can develop more effective and fair algorithms for allocating courses to students, and better meet the diverse and evolving needs of students.


\bibliographystyle{plainurl}
\bibliography{ref}
\newpage
\appendix

\section{Maximizing $b$-Matching with Conflicts}\label{b-match}
In this section, we justify our model of representing
courses as interval graphs by showing that the general
problem of assigning courses to students is NP-complete
when given arbitrary numbers of time segments (or intervals) for each
course. Namely, when each course can take place over any
arbitrary number of time periods, then the conflict 
graph can be represented as any general graph. We 
now discuss the more general problem of assigning 
resources to agents where in our specific setting, courses
can be modeled as resources and students as agents. 

One way to view the problem of maximizing the utility of assigning resources to agents, where each agent is assigned a set of non-conflicting resources, is to realize that any agent's allocation is an independent set in the conflict graph. Assigning resources to  $n$ agents then becomes a maximum graph coloring problem, where the resources have to be colored with one of $n$ different colors so that no two adjacent resources have the same color, but we simply attempt to maximize the number of colored resources (nodes). 
If the conflict graph has no restrictions or structure,  then even the simplest case becomes $NP$-hard as we show next.

\begin{figure}[htbp]
  \centering
  \includegraphics[width=0.25\textwidth]{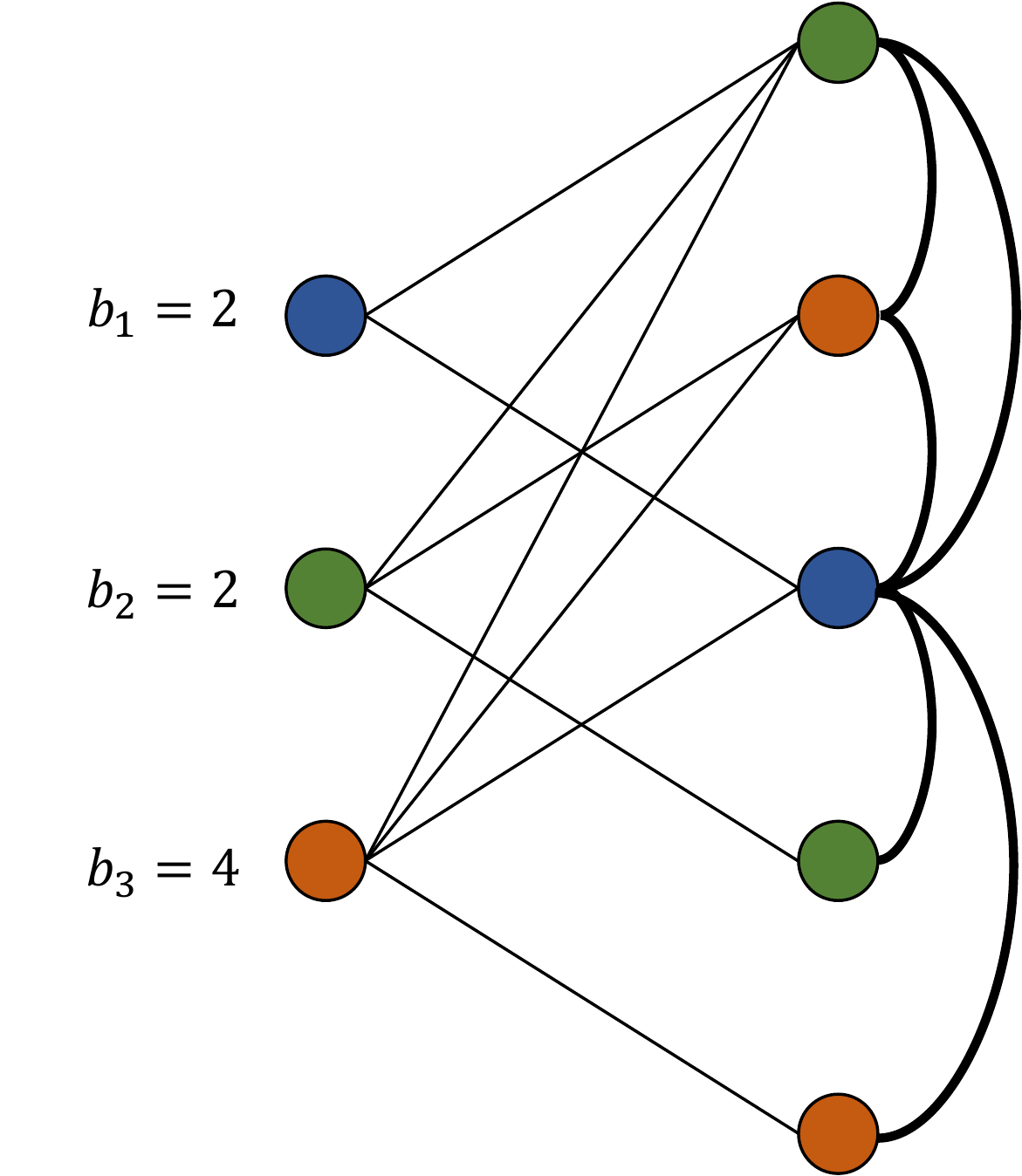}
  \caption{Example $b$-matching with allocations of resources indicated by the different colors.}
\end{figure}

A $b$-matching of any graph is a degree constrained subgraph, where the degree of any node in the subgraph cannot exceed $b(v)$, a specified value. Note that any allocation of goods to agents can be thought of as a $b$-matching where the edges encode the value of the good to that agent, and the degree constraints model the number of seats in a course (available copies of the good to be assigned to agents) and the degree constraint on the agent nodes corresponds to an upper bound as to how many resources they desire.

\begin{definition}[$b$-Matching with Conflicts (\mbmwc)]
Given a bipartite graph $G=(L \cup R,E)$, a length $|L \cup R|$ vector $\Vec{b}$ of non-negative integers, and a set of pairs $(a, a') \in F$ denoting conflicts between nodes on the same side (i.e.\ either $a \in L$ and $a' \in L$, or $a \in R$ and to $a' \in R$) such that no node $v$ can be matched to $a$ and $a'$ at the same time, a feasible b-matching with conflicts is one where the conflicts are respected and no node $p$ gets matched to more than $b(p)$ nodes on the other side. A maximum $b$-matching for \mbmwc is a feasible matching of 
maximum weight.
\end{definition}

Even if we simply want to maximize the overall weight of the $b$-matching (i.e. the sum of everyone's allocation), the problem is $NP$-hard. This can be shown by a simple reduction from independent set.

\begin{definition} [Maximum Independent Set (\mis)]
Given a graph $G=(V,E)$, set of vertices $V' \subseteq V$ is independent if and only if $\forall p, q \in V', (p,q) \notin E$, i.e. no pair of vertices in $V'$ shares an edge. A maximum independent set of a graph is an independent set with maximum cardinality.
\end{definition}

Given a graph $G$ and an integer $k$, asking for the existence of an independent set of size at least $k$ is an $NP$-complete problem. We prove the difficulty of our problem by a reduction from the Independent Set problem.

\begin{theorem}
Given a bipartite graph $G=(L \cup R,E)$, a vector $\Vec{b}$, and a set of pairs $F$ denoting conflicts, finding a  
$b$-matching satisfying \mbmwc is $NP$-hard.
\end{theorem}

\begin{proof}
Given an instance of maximum independent set problem, $G = (V, E)$, and an integer $k$, we construct an instance of \mbmwc, $H = (L \cup R, E')$ where $L$ consists of one node (agent) $v$ and $R = G$. We then create edges from $v$ to all vertices in $R$. Let $b(v) = k$ and let $b_u = 1$ for all $u \in R$. 

If we have a solution to \mbmwc in $H$ of weight $k$, then the matched vertices in $R$ give a maximum independent set in $G$ of cardinality $k$. In addition, if the graph $G$ does contain an independent set of size at least $k$ then any subset of $k$ nodes can be safely matched with $v$ (and they form a conflict free set).

\end{proof}

\section{Proofs}

\subsection{Proof of \cref{thm:uni-utility-arb-cred-cap-hardness}}\label{app:uni-utility-arb-cred-cap-hardness}
\begin{proof}
This proof is a reduction from \threepartition \cite{threepartition}.

\begin{definition}[\CAUTCDECISION]
Consider our problem \CAUTCSW in \cref{sec:model}, instead of the objective of maximizing it, the decision version of it is that given the extra parameter $k$, is there an allocation such that total student utility is $k$?
\end{definition}

\begin{definition}[(\threepartition)]

Given a multiset of numbers, can one partition the numbers into triplets such that the sum of each triplet is equal? More precisely, and with an additional restriction on each number. Given a multiset $S$ of $3m$ positive integers where $\sum_{i \in S} x_i = mT$, and each integer $x_i \in S$ satisfies $T/4 < x_i < T/2$, does there exist a partition of $S$ into $m$ disjoint subsets $S_1$, $S_2$,...,$S_m$ such that the sum of the $x_i$ values in each set $S_j$  add exactly to $T$?
\end{definition}

Given an instance of \threepartition, one can reduce it to an instance of \CAUTCDECISION where utilities are uniform and credit caps are uniform and course credit counts are arbitrary. Let there be $m$ students $s_1$, $s_2$,...,$s_m$, each with credit cap $T$, and let each number $x_i \in S$ from \threepartition represent a course of credit count $x_i$. No two courses overlap. Every student is interested in every course with uniform utilities. Let $k=mT$. If the solution to \CAUTCDECISION is yes, then the solution to \threepartition is also yes. But first, we have to prove that if there is a solution to \CAUTCDECISION, then each student is allocated exactly three courses. Since the total student utility is $k=mT$ and each student has a credit cap of $T$, each student is allocated courses whose credits sum to exactly $T$. Each student must have at least three courses, because each course $j$ has credit $c_j < T/2$. On the other hand, each student must have at most three courses, because each course $j$ has credit $c_j > T/4$. \CAUTCDECISION is therefore NP-hard.

\end{proof}

\subsection{Proof of \cref{thm:CAUTC-hardness}}\label{app:CAUTC-hardness}
\begin{proof}
This proof is based on the reduction from \arccoloring to the $k$-track assignment problem by Brucker and Nordmann \cite{k-track-ass} showing NP-hardness of the $k$-track assignment problem.

\begin{definition}[$k$-coloring problem for circular arc graphs (\arccoloring)]

Given a positive integer $k$ and a set $F$ of $n$ circular arcs ${A_1, A_2, ... A_n}$, where each $A_i$ is an ordered pair $(a_i, b_i)$ of positive integers where either $a_i < b_i$ or $b_i < a_i$, can $F$ be partitioned into $k$ disjoint subsets so that no two arcs in the same subset intersect?

\begin{figure}[h]
\centering
\begin{tikzpicture}

\draw (0,0) circle (2cm);

\draw [|-|] [domain=5:35] plot ({2.3*cos(\x)}, {2.3*sin(\x)});

\draw [|-|] [domain=70:130] plot ({2.3*cos(\x)}, {2.3*sin(\x)});

\draw [|-|] [domain=25:60] plot ({2.6*cos(\x)}, {2.6*sin(\x)});

\draw [|-|] [domain=115:240] plot ({2.6*cos(\x)}, {2.6*sin(\x)});

\draw [|-|] [domain=50:120] plot ({2.9*cos(\x)}, {2.9*sin(\x)});

\draw [|-|] [domain=95:340] plot ({3.2*cos(\x)}, {3.2*sin(\x)});


\end{tikzpicture}

\caption{A circular arc model}
\end{figure}
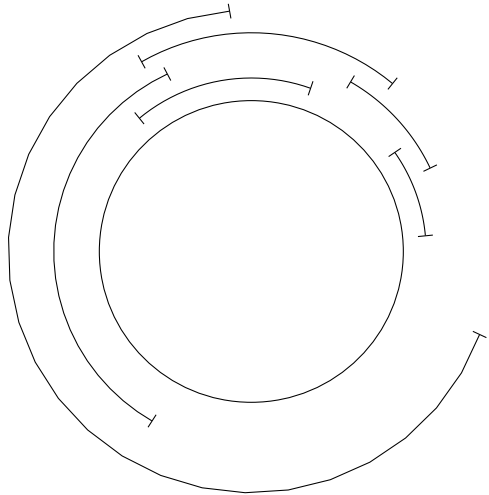

\end{definition}

The following simple reduction from \arccoloring shows that CAUTC is NP-hard: we cut the circle from the $k$-coloring problem for circular arc graphs at some arbitrary but fixed point $t$. Without loss of generality we calibrate that as $t=0$, and the courses $I_i$ have the form $I_i=[s_i, t_i]$, where each $s_i$ and $t_i$ is modulo $L$, the length of the circle. 

Now assume that only the courses $I_1,...,I_r$ contain the point $t=0$ and that $r \leq k$, for if $r > k$, then the $k$-coloring problem has no solution. We define $k$ students by making them have a utility of 1 only for the courses that overlap with the time interval $[t_j, s_j]$ for $j = 1,...,r$ and $[0,L]$ for $j = r+1,...k$. Now the problem of assigning the remaining courses $I_{r+1},...,I_n$ to these $k$ students is equivalent to the $k$-coloring problem.
\end{proof}

\subsection{Proof of \cref{thm:constant-students-DP}}\label{constant-students-DP}

\begin{proof}
We give a dynamic programming solution for two students, which is easily extendable to any constant $k$ 
number of students. We sort the courses by non-decreasing start time and use this order to consider the 
courses in our DP. We define $N(j)$ to be the set of courses that overlap with course $j$. Given an instance 
of CAUTC with a constant number of students, for each course $j \in [m]$,
course $j$ is either assigned to student $1$, to student $2$, or to no one. 
The states of our DP are as follows. For each of the two students, 
we maintain a counter, $p_1$ and $p_2$, respectively, for 
the remaining number of credits available to student $1$ and $2$; we also maintain
the set of courses available to students $1$ and $2$ where $t_1$ and $t_2$ 
denote the earliest time that a course which starts
at that time can be assigned to students
$1$ and $2$, respectively. Finally, we maintain
a counter $j$ indicating the current course being iterated on.

Each time a course $j$ is assigned to wlog student $1$, we subtract the credit count of the course, $c_j$,
from $p_1$ (the total credit count of the student course $j$ is assigned to), 
increment $t_1$ by the duration of course $j$, that is we update
$t_1$ to $t_1 + d_j$. We define our base case to be 
\begin{equation}
    OPT[p_1, p_2, t_1, t_2, m+1] = 0
\end{equation}
for any valid $p_1, p_2, t_1, t_2$ and our initial state is
\begin{equation}
    OPT[p_1, p_2, 0, 0, 0].
\end{equation}

We therefore have our recurrence scheme as follows:
\begin{equation}
    \begin{gathered}
    OPT[p_1, p_2, t_1, t_2, j] = \max(OPT[p_1, p_2, t_1, t_2, j+1], \\
\mathbbm{1}(\starttime_j \geq t_1 \cap c_j \leq p_1) \times (u_1(j) + OPT[p_1-c_j, p_2, \etime_j, t_2, j+1]), \\
\mathbbm{1}(\starttime_j \geq t_2 \cap c_j \leq p_2) \times (u_2(j) + OPT[p_1, p_2-c_j, t_1, \etime_j, j+1])
)
    \end{gathered}
\end{equation}
We now prove the optimality of our solution via induction. In the base case, course $m+1$ does not 
exist, hence, no utility  is given for the base case. We now assume for our 
induction hypothesis that the state for the $j$-th job is an optimum assignment
of courses to students for all valid values of $p_1, p_2, t_1, t_2$. 
Now, we show that the optimum solution is computed for the $(j+1)$-st job. 
For the $(j+1)$-st course, it can either be given to student $1$ or $2$ or given to no one.
Wlog suppose the $(j+1)$-st course is given to student $1$. In this case, if $\starttime_{j+1} < t_1$
or $c_j > p_1$, then the returned value is $0$ since course $j+1$ cannot be assigned to student $1$
in this case. Otherwise, we show that the states are correctly updated. When $j+1$ is assigned to 
student $1$, the amount of available credits is decreased for student $1$ by $c_{j+1}$ and
$t_1$ is increased to $\etime_{j+1}$. Since the courses are sorted in non-decreasing order by start 
time, when course $j + 1$ is being considered, no course with start time earlier than $\starttime_{j+1}$
is being considered. Thus, all courses $j' > j+1$ have start time $\geq \starttime_{j+1}$ and so 
will conflict with course $j + 1$ if and only if $\starttime_{j + 1} \leq \starttime_{j'} < \etime_{j+1}$.
Hence, setting $t_1$ to $\etime_{j+1}$ precisely eliminates the courses $j' > j+1$ that conflict with 
course $j + 1$. Since course $j + 1$ has been assigned to student $1$, the utility $u_1(j+1)$ is added. 
Finally, the counter is incremented to $j + 2$. The case for assigning $j + 1$ to student $2$ is symmetric.
When $j + 1$ is not given to either student, then no utility is added to the previous values and 
the counter is incremented to $j + 2$ with no other changes in the state. There are only three different
cases for course $j + 1$: it is assigned to either student $1$ or $2$ or assigned to no one. Using the 
induction hypothesis and taking the maximum of the three options results in the maximum value for 
assigning course $j + 1$.

Now we prove the runtime of our DP algorithm. Since
$c_j = O(1)$ for all $j \in [m]$, we can upper bound $p_1$ and $p_2$ by $O(m)$. We can bound $t_1$
and $t_2$ as follows. We only increment each of these counters to an end time of a course. There 
are at most $m$ distinct end times and thus
the total number of values $t_1$ and $t_2$ can take
is $m$. Finally, the last counter is upper bounded
by $m$. Hence, there are at most $O(m^5)$ different
unique states for our DP and our algorithm takes
$O(m^5)$ time. For $s = O(1)$ students, our 
algorithm would take $O(m^{2s + 1})$ time.
\end{proof}

\subsection{Proof of~\cref{thm:unif-creds-lengths-utilities}}\label{sec:greedy-comes-first}

\begin{proof}
We first prove the optimality of~\cref{alg:round-robin}. In this 
proof, we use the classical greedy-comes-first strategy.
In the sorted order of courses by end time, 
let $J$ be an optimum assignment
of courses to students. We show that our greedy algorithm does
not produce a worse assignment than $J$, thus proving
its optimality. We prove this via induction on the $k$-th course
in the order sorted by end time. 
We aim to 
show that for all $k \leq m$, the number of courses assigned by the
greedy algorithm to each student up to course $k$ is at least the number of courses with index $\leq k$ (in the sorted
order) assigned in $J$ to each student. 

In the base case, when $k = 1$,
no courses have been assigned yet, so either the first course is 
assigned to some student with a sufficiently large credit cap or 
no student has a sufficiently large credit cap in which case it also
cannot be assigned in $J$. We assume for our induction hypothesis that
our greedy algorithm has assigned at least as many courses up to and 
including the $k$-th course to each student as the number of courses in $J$ with index $\leq k$ (in the 
sorted order by end time) assigned
to each student. We now prove this for the $(k+1)$-st course. The trivial cases are when the $(k+1)$-st
course is not in $J$ or if the $(k+1)$-st course is assigned by the greedy 
algorithm. Let the $(k+1)$-st course be course $j$.
If the course is in $J$ and it is not assigned by the greedy algorithm to any student, then 
each student must satisfy at least one of the two following scenarios:

\begin{enumerate}
    \item Student $i \in [n]$ has not enough remaining credits.\label{case:1}
    \item Student $i \in [n]$ is assigned a conflicting course.\label{case:2}
\end{enumerate}

If \cref{case:1} is true, then student $i$ is assigned as many courses by the greedy algorithm
as they were assigned in $J$; in other words, student $i$ is assigned the maximum number of 
courses they can take; this means that the greedy algorithm returned a solution no worse than $J$, since every student has reached their credit cap (since all courses have the same number of credits), and there is no way to improve upon that.

Otherwise, if~\cref{case:1} is not true and~\cref{case:2} is true
then we consider the course with the \emph{latest} end time that is $\geq \etime_{j}$. 
Such a course must exist by our greedy algorithm since if no such conflicting course 
exists, then $j$ would be assigned to $i$. Let this conflicting course be $j'$. Then, courses $j$ and $j'$
cannot both be assigned to student $i$ in $J$. By our 
induction hypothesis, the greedy algorithm assigned at least as many courses to student $i$
with index $\leq k$ as the number of courses assigned to $i$ in $J$ with index $\leq k$.
Suppose wlog that $j'$ is the only course assigned to $i$ that conflicts with $j$
and we remove course $j'$ from student $i$'s assignment and instead assign $j$. Then, 
the number of courses assigned to $i$ cannot increase. Now we argue that removing $j'$
cannot allow another course to be assigned to $i$. Suppose there exists another
course $\ell$ that is assigned in $J$ and conflicts with $j'$ and does not conflict 
with $j$ (so that both $\ell$ and $j$ can be assigned to $i$ if $j'$ is removed).
Since all courses have the same duration, it must be the case that if $\ell$ exists
then $\ell$ has start time earlier than $j'$ and has end time earlier than the 
start time of $j$. In that case, $j'$ could not have
prevented $\ell$ from being assigned to $i$ and there exists another course assigned by greedy
to $i$ that conflicts with $\ell$. Hence, no such $\ell$ can exist and removing 
$j'$ and adding $j$ cannot lead to another course $\ell$ with start time earlier
than $\starttime_j$ to be assigned to $i$. In other words, if $\ell$ had been 
assigned to $i$ by greedy, then $j$ would also have 
been chosen, which contradicts our initial assumption that $j$ and $\ell$ conflict; 
and if $\ell$ hadn't been assigned to $i$ by greedy, it's because a course 
that starts earlier than $\ell$ overlaps with it, in which case removing $j'$ does not enable $\ell$ 
to be assigned to $i$.

Finally, courses $j'$ and $j$ cannot both be assigned to the same student in $J$. Thus,
if $j'$ is assigned to a student in $J$, then $j$ is not assigned to that student. 
Hence, we only need to consider the case when $j'$ is not assigned in $J$. 
By our argument above, at most one course in $J$ is charged
to a course assigned by our algorithm; hence,
in this case, by what we proved above and by the induction hypothesis the number of courses
assigned to $i$ by the greedy algorithm with index $\leq k+1$ 
is at least the number of courses assigned to $i$ with index $\leq k+1$ by the optimum solution $J$. 
\end{proof}

\end{document}